\documentclass{article} 
\usepackage{arxiv}

\usepackage[T1]{fontenc}
\usepackage{lmodern}

\usepackage{amsmath}
\usepackage{amssymb}
\usepackage{bbm}
\usepackage{mathtools}
\usepackage[colorlinks=true, allcolors=blue]{hyperref}
\usepackage[capitalize,noabbrev]{cleveref}
\usepackage{subfig}
\usepackage[ruled,vlined,linesnumbered]{algorithm2e}
\usepackage{todonotes}
\usepackage[capitalize,noabbrev]{cleveref}
\usepackage{tikz}
\usepackage{xcolor}
\usepackage{xparse}

\DeclareDocumentCommand\R{}{\mathbb{R}}
\DeclareDocumentCommand\onevec{o}{\IfNoValueTF{#1}{\mathbbm{1}}{\mathbbm{1}_{#1}}}
\DeclareDocumentCommand\Rnonneg{}{\R_{\geq\,0}}
\DeclareDocumentCommand\core{o}{C_{\IfValueTF{#1}{#1}{(N,c)}}}
\DeclareDocumentCommand\almostCore{o}{AC_{\IfValueTF{#1}{#1}{(N,c)}}}
\DeclareDocumentCommand\Pcore{o}{P_{\IfValueTF{#1}{#1}{(N,c)}}}

\DeclareDocumentCommand\costOfStability{o}{\text{CoS}{\IfValueTF{#1}{#1}{(N,c)}}}
\DeclareDocumentCommand\cplxP{}{\mathsf{P}}
\DeclareDocumentCommand\cplxNP{}{\mathsf{NP}}

\usepackage{amsthm}

\newtheorem{theorem}{Theorem}
\newtheorem{lemma}{Lemma}
\newtheorem{corollary}{Corollary}
\newtheorem{proposition}{Proposition}

\title{Algorithmic Solutions \\for Maximizing Shareable Costs} 

\author{
 Rong Zou \\
  Northwestern Polytechnical University\\
  Xi'an, Shaanxi, China\\
  \texttt{rongzou20@gmail.com} \\
   \And
  Boyue Lin \\
  Northwestern Polytechnical University\\
  Xi'an, Shaanxi, China\\
  \texttt{linboyuexigongda@163.com} \\
   \AND
   Marc Uetz \\
   University of Twente\\
  Enschede, The Netherlands \\
\texttt{m.uetz@utwente.nl} \\
 \And
 Matthias Walter \\
  University of Twente\\
  Enschede, The Netherlands \\
\texttt{m.walter@utwente.nl} 
}

\begin{document}

\maketitle

\begin{abstract}
This paper addresses the optimization problem to maximize the total costs that can be shared among a group of agents, while maintaining stability in the sense of the core constraints of a cooperative transferable utility game, or TU game. When maximizing total shareable costs, the cost shares must satisfy all constraints that define the core of a TU game, except for being budget balanced. The paper first gives a fairly complete picture of the computational complexity of this optimization problem, its relation to optimiztion over the core itself, and its equivalence to other, minimal core relaxations that have been proposed earlier.
We then address minimum cost spanning tree (MST) games as an example for a class of cost sharing games with non-empty core. While submodular cost functions yield efficient algorithms to maximize shareable costs, MST games have cost functions that are subadditive, but generally not submodular. Nevertheless, it is well known that cost shares in the core of MST games can be found efficiently. In contrast, we show that the maximization of shareable costs is $\cplxNP$-hard for MST games and derive a 2-approximation algorithm. Our work opens several directions for future research.
\end{abstract}




\maketitle 


\section{Introduction}
\label{sec_intro}
The fundamental algorithmic question that is addressed in this paper is: Can we maximize the total costs that can be shared among a set of agents, while maintaining coalitional stability? Here, coalitional stability refers to the core constraints  of an underlying cooperative transferable utility game, or TU game:
Any proper subset of the set of all agents has an outside option at a certain cost, and coalitional stability of cost shares means that all subsets of agents are willing to accept the cost shares, because their outside option is less attractive, meaning that it is at least as costly as the sum of their cost shares. 
This question is arguably a  fundamental question for the design of cost sharing mechanisms, and the main goal of this paper is to give more insight into its algorithmic complexity.
Several closely related results exist, but these are a bit scattered and sometimes ignorant of each other. This will be discussed in Sections~\ref{sec:intro_ac} and~\ref{sec_related_work}.

The main contributions of this paper are as follows. 
We introduce a basic polyhedral object that we refer to as the \emph{almost core} of a cooperative game.
It is obtained from the core by relaxing the requirement that the cost shares must be budget balanced. By definition, this polyhedron is non-empty.
The algorithmic problem that we address is to maximize cost shares that lie in the almost core, which is a linear optimization problem over that polyhedron.
For the case that the  core of the corresponding cooperative game is empty, it turns out that the computational problem to maximize shareable costs is equivalent to finding a minimal non-empty core relaxation in the sense of several of the core relaxations that have been proposed earlier in the literature.
This is maybe not surprising, yet a good overview of how all these relaxations relate to each other does not seem to appear anywhere in the literature.
The paper further establishes complexity theoretic results that relate computational problems for the almost core with corresponding problems for the classical core. 
While it turns out that general linear optimization over almost core and core  share the same algorithmic complexity, we show that there are classes of games where core elements can be efficiently computed, while the computation of maximal shareable costs 
cannot be done in polynomial time, unless $\cplxP=\cplxNP$.
That hardness result is obtained for a well-studied class of games with non-empty core, namely minimum cost spanning tree games.
This class of games is interesting also because the resulting cost function is subadditive but generally not submodular.
And while submodularity yields polynomial-time algorithms, our hardness result shows that subadditivity does not suffice.
For minimum cost spanning tree games, we further show how to obtain a 2-approximation algorithm for maximizing shareable costs, and we show that our analysis of that algorithm is tight.

The structure of this paper is as follows. The basic notions and definitions are given in Section~\ref{sec:intro_ac}. As previous papers have mostly focused on core relaxations for unbalanced games (that is, with an empty core), we briefly review these in Section~\ref{sec_related_work}, and discuss how they relate to the problem to compute maximal ``almost core'' cost shares. 
A novel aspect of our approach is to also address games that have a non-empty core.
Section~\ref{sec_complexity}  therefore relates linear optimization over the almost core
to the core, and we derive some algorithmic consequences.
Section~\ref{sec_MST_games} then addresses the problem to compute maximal cost shares for minimum cost spanning tree (MST) games, showing $\cplxNP$-hardness, as well as giving a 2-approximation algorithm.
We conclude with some open problems in Section~\ref{sec_conclusions}.

\section{Core and Almost Core for TU Games}
\label{sec:intro_ac}
A cooperative game with transferable utility (henceforth TU game) is described by a pair $(N,c)$ where $N=\{1,\dots,n\}$ denotes the set of agents, and $c: 2^{N}\to {\R_{\ge 0}}$ is the characteristic function that assigns to every coalition $S$ a value $c(S)$ representing the cost of an ``outside option'', which is the minimum total cost that the agents in $S$ can achieve if they cooperate amongst themselves.
With a slight overload of notation write $n=|N|$ for the number of agents.
An \emph{allocation} for $(N,c)$ is a vector $x\in \R^n$ with $x_{i}$ being the cost share allocated to agent $i\in N$. For convenience, we write $x(S)=\sum_{i\in S}x_i$.
An allocation $x$ is said to be \emph{budget balanced} if $x(N)=c(N)$.
That means that the total cost of the so-called \emph{grand coalition} $N$ is being distributed over the individual agents.
It is called stable if it satisfies \emph{coalitional stability}, i.e., $x(S) \leq c(S)$ for all $S \subsetneqq N$.
The \emph{core}~\cite{Gillies1959} of game $(N,c)$, arguably one of the most important concepts in cooperative game theory, consists of all budget balanced allocations satisfying coalitional stability.
The core of a TU game is given by
\begin{equation*}
  \core \coloneqq \{ x \in \R^n : x(S) \leq c(S) ~\forall S \subsetneqq N,~ x(N) =  c(N) \}\,.
\end{equation*}
The core of a TU game is non-empty if and only if the game is balanced~\cite{Bondareva1963,Shapley1969}.
In fact, being balanced is just a dual characterization of the non-emptiness of the polyhedron $\core$. 

When we drop the equality constraint that a core allocation is budget balanced, so do not require that $x(N) = c(N)$, it allows to vary the total cost that is distributed over the set of agents, resulting in a problem that always has a feasible solution.
This captures the idea that, depending on the underlying game, one may have to, or want to, distribute either less or more than $c(N)$.
We refer to the set of all such allocations as the \emph{almost core}, because the cost shares  fulfill almost all, i.e., all except one of the core constraints.
Formally, given a TU game $(N,c)$, define the almost core for $(N,c)$ by  
\begin{equation*}
  \almostCore \coloneqq \{ x \in \R^n : x(S) \leq c(S) ~\forall S \subsetneqq N\}\,.
\end{equation*}
Obviously, $\core \subseteq \almostCore$. The major motivation for this definition is to systematically study the algorithmic complexity of cooperative games without having to obey to budget balance, so optimization over the polyhedron $\almostCore$. Let us motivate the relevance of this problem.

On the one hand, if the total cost $c(N)$ of the grand coalition \emph{cannot} be distributed over the set of agents while maintaining coalitional stability, i.e., the game is unbalanced, it is a natural question to ask what fraction of the total cost $c(N)$ can be maximally distributed while maintaining coalitional stability.
This problem has been addressed under different names, among them the \emph{cost of stability} 
of a cooperative game~\cite{Bachrach2009}.
It has received quite some attention in the literature, e.g.~\cite{Aziz2010,Bachrach2018,Bachrach2009,Bejan2009,Bousquet2015,Gomez2011,Jain2007,Liu2009,Liu2018,Liu2022,Meir2011,Meir2013}.
Indeed, for games with empty core, maximizing $x(N)$ over the almost core is equivalent to computing the cost of stability, and also to computing some other minimal core relaxations proposed earlier in the literature; see Section~\ref{sec_related_work} for more details.

On the other hand, also if the core is non-empty one may be interested in maximizing the total cost that can be distributed over the set of agents. It reveals the maximal value for $c(N)$ that would still yield a non-empty core. One motivation for this maximization problem is to determine the maximal tax rate that could be levied on a given $c(N)$, without any subset of agents $S\subsetneqq N$ wanting to deviate. 

That said, the object of interest of this paper is the following linear program.
\begin{equation}
\label{eq:main_AC_problem}
    \max \{ x(N) : x \in \almostCore \}.
\end{equation}
The objective value of this linear program 
indicates the largest total cost that can be shared among the agents while retaining stability in the sense that no subset of agents $S \subsetneqq N$ would prefer to deviate to the outside option. We call an optimal solution value for this linear program the \emph{almost core optimum}, and any maximizer is called an \emph{optimal almost core allocation}. Sometimes we also consider the restricted problem where we also require that $x\ge 0$, which means that agents must not receive subsidies.

Clearly, the core of a game is non-empty if and only if the almost core optimum is larger than or equal to $c(N)$.
We study problem \eqref{eq:main_AC_problem} mainly for games with non-empty core, while for games with empty core we next give a fairly complete overview of its relation to earlier proposed core relaxations.

\section{Equivalent and Related Relaxations of the Core}
\label{sec_related_work}

\DeclareDocumentCommand\strongEpsilon{}{\ensuremath{\varepsilon_{\text{s}}^\star}}
\DeclareDocumentCommand\weakEpsilon{}{\ensuremath{\varepsilon_{\text{w}}^\star}}
\DeclareDocumentCommand\multEpsilon{}{\ensuremath{\varepsilon_{\text{m}}^\star}}
\DeclareDocumentCommand\approxGamma{}{\ensuremath{\gamma_{\text{a}}^\star}}
\DeclareDocumentCommand\costOfStability{}{\ensuremath{\delta^\star_{\text{CoS}}}}
\DeclareDocumentCommand\deltaExtCore{}{\ensuremath{\delta^\star_{\text{ec}}}}

We review several well-known and related concepts that were introduced in order to deal with games having an empty core and discuss their relationship to the almost core (optimum).


  The first relaxation of the core, introduced by Shapley and Shubik~\cite{Shapley1966}, is the \emph{strong $\varepsilon$-core}, defined as
  \begin{equation*}
    C^{\varepsilon}_{\text{s}}(N,c) \coloneqq \{ x \in \R^n : x(S) \leq c(S) + \varepsilon ~\forall S \subsetneqq N,~ x(N) = c(N) \}\,.
  \end{equation*}
  We denote the smallest $\varepsilon \geq 0$ for which this set is non-empty by $\strongEpsilon$.
  The corresponding set $C^{\strongEpsilon}_{\text{s}}(N,c)$ is called the \emph{least core}~\cite{Maschler1979}.

  Shapyley and Shubik~\cite{Shapley1966} also introduced the \emph{weak $\varepsilon$-core} as 
  \begin{equation*}
    C^{\varepsilon}_{\text{w}}(N,c) \coloneqq \{ x \in \R^n : x(S) \leq c(S) + \varepsilon \cdot |S| ~\forall S \subsetneqq N,~ x(N) = c(N) \}\ .
  \end{equation*}
  We denote the smallest $\varepsilon \geq 0$ for which this set is non-empty by $\weakEpsilon$. Note that by definition, for any $\varepsilon \geq 0$, $C^{\varepsilon}_{\text{s}}(N,c)\subseteq C^{\varepsilon}_{\text{w}}(N,c)$, and hence $\weakEpsilon\le \strongEpsilon$.

  Instead of using an additive relaxation of the constraints, Faigle and Kern~\cite{Faigle1993} defined the multiplicative $\varepsilon$-core as
  \begin{equation*}
    C^{\varepsilon}_{\text{m}}(N,c) \coloneqq \{ x \in \R^n : x(S) \leq (1 + \varepsilon) \cdot c(S) ~\forall S \subsetneqq N,~ x(N) = c(N) \}\,.
  \end{equation*}
  Denote the smallest $\varepsilon \geq 0$ for which this set is non-empty by~$\multEpsilon$.

  A different viewpoint is called \emph{approximate core} or \emph{$\gamma$-core}~\cite{Jain2007} for some $\gamma \in [0,1]$, it is defined as
  \begin{equation*}
    C^{\gamma}_{\text{a}}(N,c) \coloneqq \{ x \in \R^n : x(S) \leq c(S) ~\forall S \subseteq N,~ \gamma \cdot c(N) \leq x(N) \}\,.
  \end{equation*}
  Denote the largest $\gamma \leq 1$ for which this set is non-empty by $\approxGamma$.

  The gap between the almost core optimum and the total cost of the grand coalition $c(N)$ was 
  called the \emph{cost of stability} for an unbalanced cooperative game
  by Bachrach et al.~\cite{Bachrach2009}. 
  For (unbalanced) cost sharing games it is defined by Meir et al.~\cite{Meir2010} as
  \begin{equation*}
    \costOfStability \coloneqq c(N) - \max \{ x(N) : x(S) \leq c(S) ~\forall S \subseteq N \}\,.
  \end{equation*}
  An alternative viewpoint was independently introduced in a paper by Bejan and G{\'o}mez~\cite{Bejan2009} who considered, for profit sharing games, the so-called \emph{extended core}.
  In order to define it for cost sharing games, let
  \begin{equation}
     \deltaExtCore \coloneqq \min\{ t(N) : \exists (x,t) \in \R^n \times \Rnonneg^n,~ x(N) = c(N), 
     (x-t)(S) \leq c(S) ~\forall S \subsetneqq N \}\,. \label{eq_delta_ext_core}
  \end{equation}
  The \emph{extended core} is now the set of all budget balanced $x \in \R^n$, so all $x$ with $x(N) = c(N)$ for which the minimum above is attained (for suitable $t \in \Rnonneg^n$).


Yet another concept to stabilize an unbalanced game was considered by Zick, Polukarov, and Jennings~\cite{Zick2013}. Interpreting 
$t_i$ in the definition of the extended core of Bejan and G{\'o}mez~\cite{Bejan2009} 
as a discount offered to agent $i$, 
in~\cite{Zick2013} a coalitional discount $t_S$ is offered to each agent set $S$.
This is an exponential blowup of the solution space, which however gives more flexibility.

For unbalanced games, i.e., games with an empty core, computing the almost core optimum is clearly the same as computing the cost of stability $\costOfStability$.
The following theorem further shows how the different core relaxations are related with each other with respect to optimization.
Some of these relations were known, e.g., Liu et al.~\cite{Liu2016,Liu2018} mention the equivalence of computing \approxGamma\ and \costOfStability\ and \multEpsilon, and also the relation between the cost of stability \costOfStability\  and the smallest weak $\varepsilon$-core \weakEpsilon\ appears in \cite{Bachrach2009,Meir2011}, 
yet we are not aware of a summarizing overview of how the different relaxations relate.
Hence we give this summary here for the sake of completeness, and also give the short proof.
\begin{theorem}[in parts folklore]
  \label{thm_equivalence}
  For any TU game $(N,c)$ with empty core, the optimization problems for the weak $\varepsilon$-core, the multiplicative $\varepsilon$-core, the cost of stability and the extended core are equivalent.
  In particular, the values satisfy
  \begin{align*}
       \deltaExtCore
    = (1 - \approxGamma) \cdot c(N)
    = \frac{ \multEpsilon }{ 1 + \multEpsilon } \cdot c(N)
    = \costOfStability
    = \weakEpsilon \cdot n\,.
  \end{align*}
\end{theorem}
\begin{proof}
  First, we establish $\costOfStability = \deltaExtCore$.
  We substitute $x - t$ by $x'$ in~\eqref{eq_delta_ext_core} and obtain
  \begin{multline*}
    \deltaExtCore = \min\{ t(N) : \exists (x',t) \in \R^n \times \Rnonneg^n,~ x'(N) + t(N) = c(N),
    ~x'(S) \leq c(S) ~\forall S \subsetneqq N \}\,.
  \end{multline*}
  Now it is easy to see that the actual entries of $t$ do not matter (except for nonnegativity), but only the value $t(N)$ is important.
  This yields $\costOfStability = \deltaExtCore$.

  Second, we show $\costOfStability = (1 - \approxGamma) \cdot c(N)$.
  To this end, observe
  \begin{equation*}
    \approxGamma = \max \{ \gamma \in \R : \exists x \in \R^n,~ x(S) \leq c(S) ~\forall S \subseteq N,~ x(N) = \gamma c(N) \}\,.
  \end{equation*}
  Clearly, the maximum is attained by $x^\star \in \R^n$ with $x^\star(N)$ maximum.
  Moreover, the value of $\approxGamma$ is then equal to $x^\star(N) / c(N)$.
  This shows $\costOfStability / c(N) = 1 - \approxGamma$.

  Third, we show $1 - \approxGamma = \multEpsilon / (1+\multEpsilon)$.
  Observe that the map $\pi : \R^n \to \R^n$ defined by $\pi(x) = (1+\varepsilon) x$ induces a bijection between allocations $x \in \R^n$ with $x(S) \leq c(S)$ for all $S \subseteq N$ and allocations $\pi(x)$ with $\pi(x)(S) \leq (1+\varepsilon) c(S)$ for all $S \subseteq N$.
  Moreover, $\pi(x)(N)=(1+\varepsilon) x(N)$.
  Hence, $C^{\varepsilon}_{\text{m}}(N,c)$ is (non-)empty if and only if $C^{\gamma}_{\text{a}}(N,c)$ is (non-)empty, where $\gamma = 1 / (1+\varepsilon)$ holds.
  This implies $\approxGamma = 1 / (1 + \multEpsilon)$.
  
  We finally show $\costOfStability = \weakEpsilon \cdot n$.
  To this end, in
  \[
     \weakEpsilon = \min \{ \varepsilon \geq 0: \exists x, ~x(S) \leq c(S) + \varepsilon \cdot |S| ~\forall S \subsetneqq N,~ x(N) = c(N) \}
  \]
  we substitute $x$ by $x' + (\varepsilon, \varepsilon, \dotsc, \varepsilon)$ which yields
  \[
    \weakEpsilon = \min \{ \varepsilon \geq 0: \exists x', ~x'(S) \leq c(S) ~\forall S \subsetneqq N,~ x'(N) + \varepsilon \cdot n = c(N) \}\,.
  \]
  Clearly, the minimum $\weakEpsilon$ is attained if and only if $\varepsilon \cdot n = \costOfStability$ holds.
\end{proof}

Moreover, it was shown in~\cite[Section 4]{Meir2011} that $\weakEpsilon \ge \frac{1}{n-1}\strongEpsilon$.
Further relations between the cost of stability $\costOfStability$ and other core relaxations for specific classes of games appear in~\cite{Meir2011,Bachrach2018}.
For instance, it is true that for supperadditive (profit sharing) games, $\costOfStability\le \sqrt{n}\strongEpsilon$ and $\sqrt{n}\weakEpsilon\le \strongEpsilon$.
Indeed, much of the previous work in this direction was about determining bounds on the cost of stability~\cite{Bachrach2009,Meir2010,Meir2011,Meir2013,Bousquet2015} or other structural insights~\cite{Bejan2009}.

However,  algorithmic considerations were also made for specific (unbalanced) games such as showing hardness for computing the price of stability of general weighted voting games~\cite{Bachrach2009}, and showing hardness for computing the price of stability of threshold network flow games~\cite{Resnick2009}. 
Moreover, Aziz, Brandt and Harrenstein~\cite{Aziz2010} give several results on the computational complexity of computing the cost of stability (and other measures) for several combinatorial games such as weighted graph, voting or matching games, as well as their threshold versions. 

One of the few papers which considers the impact of restrictions on possible coalition formation in relation to algorithmic questions, and in that respect also related to our work, is by Chalkiadakis, Greco and Markakis~\cite{Chalkiadakis2016}.
In the spirit of Myerson~\cite{Myerson1977}, they assume that the formation of coalitions is restricted by a so-called interaction graph, and analyze how the computational complexity of several core-related concepts such as core membership, core non-emptiness, or cost of stability depends on the structure of that graph.
Under different assumptions on polynomial-time compact representations of the underlying game, their results include hardness as well as tractability results that depend on the interaction graph.
Their results also imply hardness of computing the cost of stability for \emph{arbitrary} subadditive (cost) games.

Also approximations of $\multEpsilon$ for the multiplicative $(1+\varepsilon)$-core and corresponding allocations have been obtained, e.g., for the symmetric traveling salesperson game by Faigle et al.~\cite{Faigle1998}, and for the asymmetric case also by Bl\"aser et al.~\cite{Blaser2008}.

There are papers that attack the problem from a mathematical optimization and computational point of view.
Under the name ``optimal cost share problem'' (OCSP), Caprara and Letchford~\cite{Caprara2010} suggest how to obtain $\gamma$-core solutions for a generalization of certain combinatorial games, named integer minimization games, using column or row generation techniques.
Under the name ``optimal cost allocation problem'' (OCAP), also Liu, Qi and Xu~\cite{Liu2016} follow the line of research initiated by~\cite{Caprara2010} and give computational results using Lagrangian relaxations.
A related line of research~\cite{Liu2018} is to consider the strong $\varepsilon$-core relaxation parameterized by $\varepsilon$ as given by the function
\begin{equation*}
  \omega(\varepsilon) \coloneqq \min_{x \in \R^n}\{ c(N)-x(N) : x(S) \leq c(S) + \varepsilon ~\forall S \subsetneqq N \}\,,
  \end{equation*}
and to approximate it computationally.
This so-called ``penalty-subsidy function''~\cite{Liu2018} is further studied in another variant in a follow up paper~\cite{Liu2022}, there approximating it using Langragian relaxation techniques, and with computational results specifically for traveling salesperson games.

Also the problem to compute allocations in the least core has been considered in the literature.
For cooperative games with submodular cost functions, it can be computed in polynomial time~\cite{Deng1998}, while for supermodular cost cooperative games it is $\cplxNP$-hard to compute, and even hard to approximate~\cite{SchulzU2010}. 

Specifically relevant for our work are results by Faigle et al.~\cite{Faigle2000} who show, among other things, $\cplxNP$-hardness to compute a cost allocation in the so-called $f$-least core for minimum cost spanning tree games, which is a tightening of the core constraints to $x(S) \leq c(S) -\varepsilon f(S)$ for certain non-negative functions $f$.
As we will argue in Section~\ref{sec_MST_games}, their result also implies hardness of computing optimal almost core allocations for the class of minimum cost spanning tree games. 

Finally, we briefly discuss the relation of almost core allocations to cost sharing methods in mechanism design. First, note that cost sharing methods  with an additional property called population monotonicity have been considered~\cite{Sprumont}, specifically also for spanning tree games~\cite{KentSkorin-Kapov1996,TijsEtAL2004}.  
Population monotonicity, also called cross monotonicity, means that one computes cost shares $x^Q$ in the cores of all subgames induced by subsets of agents $Q\subseteq N$, and the cost shares per agent $x_i^Q$ are monotonically non-increasing in $Q$.
Cross monotonicity is desirable because it yields group-strategyproof mechanisms in a setting where the individual costs of agents are private information~\cite{Moulin}.
In that context, K\"onemann et al.\ \cite{KoeneEtAl2008} show how to obtain a cross-monotonic cost sharing method for the Steiner forest problem with the property of being 2-budget balanced. That means that one computes, for all $Q\subseteq N$, a Steiner forest for agents $Q$, with cost $\tilde{c}(Q)$,
along with cost shares $x^Q$ so that $\frac12 \tilde{c}(Q) \le x^Q(Q)\le c(Q)$. 
As a byproduct, this yields an allocation $x^N$ in the corresponding almost core of Steiner forest games, with $x^N(S)\le c(S)$, for all $S\subseteq N$ and the additional property that $x^N(N)\ge\frac12 \tilde{c}(N)$.
However, the distinguishing feature in this paper, namely that the constraint $x(N)\le c(N)$ is absent and allocations with $x(N)\not\le c(N)$ are allowed, is not present in those works.

\section{Computational Complexity Considerations}
\label{sec_complexity}

In this section we investigate the computational complexity of optimization problems related to the (nonnegative) core and almost core.
To capture results for the general and the nonnegative case, we consider linear optimization over the polyhedra
\begin{align*}
  \almostCore \quad\text{ and }\quad
  \Pcore{} &\coloneqq \{ x \in \R^n : x(S) \leq c(S) ~\forall S \subseteq N \}.
\end{align*}
as well as optimization over $\Pcore \cap \Rnonneg^n$ and $\almostCore \cap \Rnonneg^n$ for families of games $(N,c)$.
Note that if the core is non-empty then it is the set of optimal solutions when maximizing $\onevec\cdot x$ over $\Pcore$.
Also note that whenever the core of a game $(N,c)$ is empty, this means that the constraint $x(N)\le c(N)$ is implied by the set of constraints $x(S)\le c(S)$, $S\subsetneqq N$,
which in turn implies $\Pcore=\almostCore$.
For games with non-empty core, we get the following correspondence between the optimization problems for the two polyhedra. The proof goes by showing the corresponding separation problems are polynomially equivalent.
\begin{theorem}
  \label{thm_opt_general}
  For a family of games $(N,c)$,
  linear optimization problems over $\almostCore$ can be solved in polynomial time if and only if linear optimization problems over $\Pcore$ can be solved in polynomial time.
\end{theorem}
\begin{proof}
  In order to prove the result we make use of the equivalence of optimization and separation~\cite{GroetschelLS81,KarpP80,PadbergR81}.
  This means, we show that we can solve the separation problem for $\Pcore$ if and only if we can solve the separation problem for $\almostCore$.
  Since $\Pcore = \{ x \in \almostCore : x(N) \leq c(N) \}$ holds, separation over $\Pcore$ reduces to separation over $\almostCore$ plus an explicit check of a single inequality.

  Hence, it remains to show how to solve the separation problem for $\almostCore$.
  For given $\hat{x} \in \R^n$, we construct $n$ points $\hat{x}^k \in \R^n$ ($k = 1,2,\dots,n$) which are copies of $\hat{x}$ except for $\hat{x}^k_k \coloneqq \min(\hat{x}_k, c(N) - \sum_{i \in N \setminus \{k\}} \hat{x}_i)$.
  Note that by construction $\hat{x}^k \leq \hat{x}$ and $\hat{x}^k(N) \leq c(N)$ hold.
  We then query a separation oracle of $\Pcore$ with each $\hat{x}^k$.

  Suppose such a query yields $\hat{x}^k(S) > c(S)$ for some $S \subseteq N$.
  Due to $\hat{x}^k(N) \leq c(N)$ we have $S \neq N$.
  Moreover, $\hat{x} \geq \hat{x}^k$ implies $\hat{x}(S) > c(S)$, and we can return the same violated inequality.

  Otherwise, we have $\hat{x}^k \in \Pcore$ for all $k \in N$ and claim $\hat{x} \in \almostCore$.
  To prove this claim we assume that, for the sake of contradiction, $\hat{x}(S) > c(S)$ holds for some $S \subsetneqq N$.
  Let $k \in N \setminus S$.
  Since $\hat{x}^k_i = \hat{x}_i$ holds for all $i \in S$, we have $\hat{x}^k(S) = \hat{x}(S) > c(S)$.
  This contradicts the fact that $\hat{x}^k \in \Pcore$ holds.
\end{proof}

It turns out that almost the same result is true when we also require that there are no subsidies, that is $x\ge 0$.
For linking the non-negative core to the non-negative almost core, it requires an assumption on the characteristic function.
\begin{equation}
  c(N\setminus\{k\}) \leq c(N) \qquad \forall k \in N.
  \label{eq_last_monotone}
\end{equation}
This condition holds, for instance, for monotone functions $c$, and implies that the core is contained in $\Rnonneg^n$; see Lemma~2 and Theorem~1 in~\cite{DrechselK10}.

\begin{theorem}
  \label{thm_opt_last_monotone}
  For a family of games $(N,c)$ satisfying~\eqref{eq_last_monotone}, linear optimization problems over $\almostCore \cap \Rnonneg^n$ can be solved in polynomial time if and only if linear optimization problems over $\Pcore \cap \Rnonneg^n$ can be solved in polynomial time.
\end{theorem}

The proof is a rather straightforward extension of that of Theorem~\ref{thm_opt_general}, additionally making use of condition \eqref{eq_last_monotone} to guarantee nonnegativity. We obtain an immediate consequence from these two theorems.

\begin{corollary}\label{cor:submodular_ac_problem}
  For a family of games $(N,c)$ for which $c(\,\cdot\,)$ is submodular (and~\eqref{eq_last_monotone} holds) one can find a (non-negative) optimal almost core allocation in polynomial time.
\end{corollary}

\begin{proof}
  For submodular $c(\,\cdot\,)$ one can optimize any linear objective function over $\Pcore$ using the Greedy algorithm~\cite{Edmonds1970}.  The result follows from \cref{thm_opt_general,thm_opt_last_monotone}.
\end{proof}

One might wonder if for submodular $c(\,\cdot\,)$, one can find a combinatorial algorithm to compute an optimal almost core allocation, instead of relying on the equivalence of optimization and separation. After all, for submodular $c(\,\cdot\,)$, $\Pcore$ is a polymatroid, so $\almostCore$ is a polymatroid minus the single linear constraint $x(N)\le c(N)$. We leave this as an open problem.

The above results only make statements about optimizing arbitrary objective vectors over these polyhedra.
In particular we cannot draw conclusions about hardness of the computation of an optimal almost core allocation, which is maximizing $\onevec\cdot x$ over $\almostCore$.
However, it is easy to see that this problem cannot be easier than deciding non-emptiness of the core, as the core of a game $(N,c)$ is non-empty if and only if the almost core optimum is at least $c(N)$.
Hence we immediately get the following.

\begin{theorem}\label{thm-hardness-via-emptiness}
  Consider a family of games $(N,c)$ for which deciding (non-)emptiness of the core is $\cplxNP$-hard.
  Then an efficient algorithm to compute an optimal almost core allocation cannot exist, unless $\cplxP = \cplxNP$.
\end{theorem}

It is well known that there exist games for which it is $\cplxNP$-hard to decide non-emptiness of the core, e.g., weighted graph games~\cite{deng1994}, or unrooted metric traveling salesperson games~\cite{Caprara2010}.
Hence, we cannot hope for a polynomial-time algorithm that computes an optimal almost core allocation for arbitrary games.

In contrast, the maximization of $x(N)$ becomes trivial for games $(N,c)$ with superadditive characteristic function $c(\,\cdot\,)$, as the set of constraints $x(\{i\})\le c(\{i\})$, $i=1,\dots,n$, already imply all other constraints $x(S)\le c(S)$, $S\subseteqq N$, and one can simply define $x_i\coloneqq c(\{i\})$.
In particular, $x(N)\le c(N)$ is implied and $\Pcore = \almostCore$.
Generalizing, the same is true for classes of games where a polynomial number of constraints can be shown to be sufficient to define the complete core.
As an example we mention \emph{matching games} in undirected graphs \cite{KernP2003}, where the core is described by the polynomially many core constraints induced by all edges of the underlying graph, as these can be shown to imply all other core constraints. 
\begin{proposition}
\label{prop_superadditive}
  Whenever $\Pcore$ is described by a polynomial number of inequalities, finding an optimal (almost) core allocation can be done in polynomial time by linear programming. 
\end{proposition}

Note that \cref{prop_superadditive} includes supermodular cost functions.
It is therefore interesting to note that for supermodular cost games, it is $\cplxNP$-hard to approximate the least core value $\strongEpsilon$ better than a factor $17/16$~\cite{SchulzU2010}.
The reason for this discrepancy is the simple fact that the least core is based on the strong $\varepsilon$-core, while the almost core relates to the weak $\varepsilon$-core as per Theorem~\ref{thm_equivalence}.

It also turns out that condition~\eqref{eq_last_monotone} implies that the value of an almost core allocation cannot exceed that of a core allocation by much.
\begin{proposition}
  \label{thm_last_monotone_approximation}
  Let $(N,c)$ be a game that satisfies~\eqref{eq_last_monotone}.
  Then every $x \in \almostCore$ satisfies
  \begin{equation*}
    x(N) \leq \big(1 + \tfrac{1}{n-1} \big) c(N)\ .
  \end{equation*}
\end{proposition}

\begin{proof}
  Let $x \in \almostCore$.
  We obtain
  \begin{align*}
    (n-1) \cdot x(N)
    =& \sum_{k \in N} x(N \setminus \{k\}) \\
    \leq& \sum_{k \in N}  c(N \setminus \{k\})
    \leq \sum_{k \in N} c(N)
    = n \cdot c(N),
  \end{align*}
  where the first inequality follows from feasibility of $x$ and the second follows from~\eqref{eq_last_monotone}.
\end{proof}

Condition~\eqref{eq_last_monotone} implies non-negativity for all core allocations and all optimal almost core allocations, as $x(N) \geq c(N)$, so $x_i \geq c(N) - c(N\setminus\{i\}) \geq 0$ for all $i\in N$.
However, this does not mean that a non-negativity requirement implies that the almost core optimum is close to $c(N)$.
In the next section, we will see that this gap can be arbitrarily large (see \cref{thm_mst_bad_approximation}).

\section{Minimum Cost Spanning Tree Games and Approximation}
\label{sec_MST_games}

In this section we address a well known special class of games known as minimum cost spanning tree (MST) games~\cite{ClausK1973,Bird1976,Granot1981}.
In MST games, the agents are nodes in an edge-weighted undirected graph $G$, and the cost of the outside option for a set of agents $S$ is determined by the cost of a minimum cost spanning tree in the subgraph induced by these agents.
MST games are known to have a non-empty core.
Moreover, it is known that finding some element in the core is computationally easy and can be done by computing a minimum cost spanning tree~\cite{Granot1981}.
The optimization problem that we address asks for the maximal amount that can be charged to the agents while no proper subset of agents would prefer the outside option. 

While in general, maximizing shareable costs is the same as asking for the maximum value $c(N)$ that still yields a non-empty core, the question may appear a bit artificial for minimum cost spanning tree games, as there, the value $c(N)$ is computed as the minimum cost of a spanning tree for all agents $N$. 
From a practical viewpoint this can be motivated by assuming there are exogenous physical or legal restrictions that prohibit the grand coalition $N$ to form, so that the player set $N$ has no bargaining power.

Apart from that, there is a more theoretical perspective that motivates studying the almost core for MST games.
One can easily see that for MST games the cost function $c(\,\cdot\,)$ is subadditive, yet it is not submodular in general; see \cite{KohSanita2020} for a characterization when it is.
Recalling that the computation of maximum shareable costs can be done in polynomial time when $c(\,\cdot\,)$ is submodular, it is a natural question to ask if this still holds for subadditive cost functions.
In that respect, note that the weighted graph games as studied in~\cite{deng1994} have polynomial-time algorithms to decide non-emptiness of the core whenever $c(\,\cdot\,)$ is subadditive.
Hence Theorem~\ref{thm-hardness-via-emptiness} applied to weighted graph  games does not give an answer to this question.
Also the results of~\cite{Chalkiadakis2016} yield that
there exist subadditive cost games for which the computation of the price of stability, hence computation of the almost core optimum is hard, yet that result also relies on hardness of the problem to decide non-emptiness of the core.
We next show that even for MST games with monotone cost function $c(\,\cdot\,)$, despite always having a non-empty core, maximizing shareable costs cannot be done efficiently unless $\cplxP = \cplxNP$.





\subsection{Preliminaries}
Let us first formally define the problem and recall what is known. We are given an edge-weighted, undirected graph $G=(N\cup\{0\},E)$ with non-negative edge weights $w:E\to \R_{\ge 0}$, where node $0$ is a special node referred to as ``supplier'' node.
Without loss of generality we may assume that the graph is complete by adding dummy edges with large enough cost.
The agents of the game are the vertices $N$ of the graph, and the characteristic function of the game is given by  minimum cost spanning trees.
That is, the cost of any subset of vertices $S\subseteq N$ is defined as the cost of a minimum cost spanning tree on the subgraph induced by vertex set $S\cup\{0\}$. So if we let 
$\mathcal{T}(S)$ be the set of spanning trees for the subgraph induced by vertex set $S\cup\{0\}$, then the characteristic function is defined as:
\[
  c(S) \coloneqq \min_{T\in\mathcal{T}(S)} \left\{ 
  w(T)
  \right\}\,.
\]

Following~\cite{Granot1981}, the associated monotonized minimum cost spanning tree game $(N,\bar{c})$ is obtained by defining the characteristic function using the monotonized cost function 
$
\bar{c}(S)\coloneqq \min_{R\supseteq S}c(R)\,.
$  
This is motivated by assuming that agents can also use other agents as ``Steiner nodes''.
Indeed, note that $\bar{c}(S)\le \bar{c}(R)$ for $S\subseteq R$, and for the associated cores of these two games, we have that $\core[(N,\bar{c})] \subseteq \core$.
Moreover, it is well known that the core of both games is non-empty, and a core allocation $x \in \core[(N,\bar{c})]$ is obtained in polynomial time by just one minimum cost spanning tree computation: if $T$ is some MST, let $e_v\in T$ be the edge incident with $v$ on the unique path from $v$ to the supplier node $0$ in $T$, then letting
\[
  x_v \coloneqq w(e_v)\,,
\]
one gets an element $x$ in the core of the monotonized minimum cost spanning tree game $(N,\bar{c})$~\cite{Granot1981}, and hence also a core element for the game $(N,c)$.
One convenient way of thinking about this core allocation is a run of Prim's algorithm to compute minimum cost spanning trees~\cite{Prim57}: 
starting to build the tree with vertex $0$, whenever a new vertex $v$ is added to the spanning tree constructed so far, $v$ gets charged the cost of the edge $e_v$ that connects $v$.

In summary, computing \emph{some} core allocation can be done efficiently, while linear optimization over the core of MST games is co-NP hard (under Turing reductions)~\cite{Faigle1997}.
We are interested in the same questions but for the case that the budget balance constraint is absent. So we seek
solutions to the almost core maximization problem
\begin{equation}
  \label{eq:MSTAlmostCoreMaxProblem_general}
  \max x(N)\ s.t.\ x\in\almostCore\,, 
\end{equation}
when $c(\cdot)$ is the characteristic function defined by minimum cost spanning trees. 
The interpretation of the lacking constraint $x(N)=c(N)$ is that the grand coalition cannot establish the solution with cost $c(N)$ on its own, say by legal restrictions.

\subsection{Computational Complexity}

As a first result, and not surprising, linear optimization over the almost core is hard for MST games.
\begin{corollary}
  \label{thm_opt_mst_hard}
  For minimum cost spanning tree games $(N,c)$, a polynomial-time algorithm for linear optimization over $\almostCore$ 
  would yield \textup{P}$=$\textup{NP}.
\end{corollary}

\begin{proof}
  The result follows from \cref{thm_opt_general} and the fact that the membership problem for the core of $(N,c)$ is a co$\cplxNP$-hard problem for MST games~\cite{Faigle1997}.
  \qed
\end{proof}

What is more interesting is that optimizing $\onevec\cdot x$ over the almost core remains hard for MST games.
\begin{theorem}
  \label{prop_hardness_MSt_AC}
  Computing an optimal almost core allocation in \eqref{eq:MSTAlmostCoreMaxProblem_general} for minimum cost spanning tree games is $\cplxNP$-hard, and this is also true for monotonized minimum cost spanning tree games $(N,\bar{c})$.
\end{theorem}

\begin{proof}
Let $\varepsilon^\star$ be the largest $\varepsilon$ for which the linear inequality system
\begin{align}
\label{eq_faigle_system}
  x(S) \leq (1-\varepsilon)c(S)  ~\forall S \subsetneqq N, \quad x(N) = c(N)
\end{align}
has a solution.
In~\cite{Faigle2000} it is shown that finding a feasible solution $x$ for \eqref{eq_faigle_system} with respect to $\varepsilon^\star$ is $\cplxNP$-hard. Note that in the reduction leading to this hardness result, $c(N)>0$.
Then, given an optimum almost core allocation $x^{\text{AC}}$,  $x^{\text{AC}}(N)\ge c(N)>0$, and we can obtain $\varepsilon^\star \coloneqq 1 - c(N)/x^{\text{AC}}(N)$.  
It is now easy to see that the vector $x' \coloneqq (1-\varepsilon^\star) x^{\text{AC}}$ is a feasible solution for \eqref{eq_faigle_system}. 
To see that the so-defined $\varepsilon^\star$ is indeed maximal, observe that scaling any feasible vector in \eqref{eq_faigle_system} by $1/(1-\varepsilon^\star)$ yields an almost core allocation.
Hence, computation of an almost core optimum for MST games yields a solution for an $\cplxNP$-hard problem.
To see the last claim about monotonized minimum cost spanning tree games, observe that the underlying reduction from the $\cplxNP$-hard minimum cover problem in~\cite{Faigle2000} yields a minimum cost spanning tree game that has a monotone cost function $c(\cdot)$ by definition.
\end{proof}

Next, we note that in general, the almost core optimum may be arbitrarily larger than $c(N)$ for MST games. This is remarkable in view of \cref{thm_last_monotone_approximation}, which shows that under condition \eqref{eq_last_monotone}, any core allocation yields a good approximation for an optimal almost core allocation, as they differ by a factor at most $n/(n-1)$.
A fortiori, the same holds for the monotonized MST games $(N,\bar{c})$.
For general MST games $(N,c)$, and without condition \eqref{eq_last_monotone}, this gap can be large.
\begin{proposition}
  \label{thm_mst_bad_approximation}
  The almost core optimum can be arbitrarily larger than $c(N)$, even for minimum cost spanning tree games and when we require that $x\ge 0$.
\end{proposition}

\begin{proof}
Consider the instance depicted in Figure~\ref{fig:large_gap}, 
for some value $k>0$.
Then $c(N)=0$, while $x=(0,0,k)$ is an optimal non-negative  almost core allocation with value $k$.
\end{proof}

In the following we consider problem \eqref{eq:MSTAlmostCoreMaxProblem_general} but with the added constraint that $x\ge 0$.
\begin{equation}
  \label{eq:MSTAlmostCoreMaxProblem}
  \max x(N)\ s.t.\ x\in\almostCore \text{, and} \ x\ge 0\,.
\end{equation}
The presence of the constraint $x\ge 0$ means that agents must not be subsidized. As can be seen from the example in Figure~\ref{fig:subsidy_example}, such subsidies may indeed be necessary in the optimal solution to \eqref{eq:MSTAlmostCoreMaxProblem_general}, as there, the only optimal solution uses cost shares $(-k,k,k)$.
However, we next show that such subsidies are in some sense an artifact of ``small edge costs'' in graph $G$, and for optimization purposes can be neglected in the following sense.

\begin{theorem}
  \label{thm:hardness_MSt_AC_nonneg}
  Every instance of the almost core optimization problem \eqref{eq:MSTAlmostCoreMaxProblem_general} can be reduced in polynomial time to an instance of problem \eqref{eq:MSTAlmostCoreMaxProblem}.
  Consequently, problem~\eqref{eq:MSTAlmostCoreMaxProblem} is also $\cplxNP$-hard for minimum cost spanning tree games.
\end{theorem}
As a matter of fact, the $\cplxNP$-hardness also follows from the fact that for monotonized MST games, we have that $x\ge 0$ is redundant in \eqref{eq:MSTAlmostCoreMaxProblem},
recalling that  $x_i\ge c(N)-c(N\setminus\{i\})\ge 0$ for all $i\in N$.
\begin{proof}
  Given an instance of \eqref{eq:MSTAlmostCoreMaxProblem_general} with edge costs $w$, define new edge costs $w'(e) \coloneqq w(e)+M$, $e\in E$, for large enough constant $M$ to be defined later. Note that $c'(S) = c(S)+|S|\cdot M$ for $S\subseteq N$.
  We argue that this renders $x \geq 0$ redundant.
  Consider an optimal solution $x'$ to problem \eqref{eq:MSTAlmostCoreMaxProblem} for edge costs $w'$, and define $x \coloneqq x' - (M,\dotsc,M)$.
  Now we have $x(S) = x'(S) - |S| \cdot M \leq c'(S) - |S|\cdot M = c(S)$ for all $S \subsetneqq N$, so $x$ is feasible for problem \eqref{eq:MSTAlmostCoreMaxProblem_general} with edge costs $w$.
  We show that $x$ must be optimal for problem \eqref{eq:MSTAlmostCoreMaxProblem_general}.
  Considering any solution $y$ \emph{optimal} for \eqref{eq:MSTAlmostCoreMaxProblem_general}, there exists a number $M$ that can be computed in polynomial time so that $y' \coloneqq y + (M,\dotsc,M) \geq 0$, and $y'(S) = y(S) + |S| \cdot M \leq c(S) + |S| \cdot M = c'(S)$, so $y'$ is feasible for \eqref{eq:MSTAlmostCoreMaxProblem} with cost function $w'$.
  Hence $y(N) > x(N)$ yields the contradiction  $y'(N) > x'(N)$.
  To argue about $M$, observe that in \eqref{eq:MSTAlmostCoreMaxProblem_general} we {maximize} $x(N)$, hence for any {optimal} solution $y$ in \eqref{eq:MSTAlmostCoreMaxProblem_general}, and any $i\in N$ there exists $S \ni i$  so that  $y(S) = c(S)$,
  hence $y_i=c(S)-\sum_{j\in S, j\ne i}y_j\ge -\sum_{j\in N}c(\{j\})$, where the last inequality holds because $c(S)\ge 0$, $y_j \leq c(\{j\})$ for all $j \in N$, and $c(\{j\})=w(\{0,j\})\ge 0$. In other words, letting $M \coloneqq  \sum_{j \in N} c(\{j\})$ suffices so that $y_i \geq -M$  for all $i \in N$, and hence $y'\ge 0$ as required.
\end{proof}

{\sc Remark.} The above reduction of computing arbitrary allocations to computing non-negative allocations generalizes to all cost sharing games $(N,c)$ by defining $c'(S)\coloneqq(S)+|S|\cdot M$ for all subsets $S\subsetneqq N$.

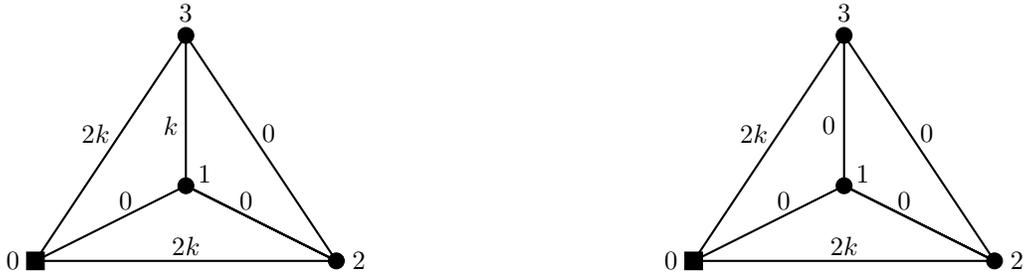
\begin{figure}[ht]
  \subfloat[MST game with unbounded relative gap between $c(N)$ and almost core optimum.\label{fig:large_gap}]{\begin{minipage}{0.47\textwidth}\centering
\begin{tikzpicture}
\node[shape=rectangle,draw=black,fill=black] at (0,0) {};
\draw[fill=black] (4,0) circle (3pt);
\draw[fill=black] (2,1) circle (3pt);
\draw[fill=black] (2,3) circle (3pt);
\node at (-0.3,0) {0};
\node at (4.3,0) {2};
\node at (2.25,1.15) {1};
\node at (2,3.3) {3};
\node at (2,0.2) {$2k$};
\node at (0.8,1.7) {$2k$};
\node at (1.2,0.8) {$0$};
\node at (2.8,0.8) {$0$};
\node at (1.8,1.8) {$k$};
\node at (3.1,1.7) {$0$};
\draw[thick] (0,0) -- (4,0) -- (2,1) -- (0,0) -- (2,3) -- (4,0) -- (2,1) -- (2,3);
\end{tikzpicture}
\end{minipage}}
  \hfill
  \subfloat[MST game where subsidies are necessary for the almost core optimum.\label{fig:subsidy_example}]{\begin{minipage}{0.47\textwidth}\centering
  \begin{tikzpicture}
\node[shape=rectangle,draw=black,fill=black] at (0,0) {};
\draw[fill=black] (4,0) circle (3pt);
\draw[fill=black] (2,1) circle (3pt);
\draw[fill=black] (2,3) circle (3pt);
\node at (-0.3,0) {0};
\node at (4.3,0) {2};
\node at (2.25,1.15) {1};
\node at (2,3.3) {3};
\node at (2,0.2) {$2k$};
\node at (0.8,1.7) {$2k$};
\node at (1.2,0.8) {$0$};
\node at (2.8,0.8) {$0$};
\node at (1.8,1.8) {$0$};
\node at (3.1,1.7) {$0$};
\draw[thick] (0,0) -- (4,0) -- (2,1) -- (0,0) -- (2,3) -- (4,0) -- (2,1) -- (2,3);
\end{tikzpicture}
  \end{minipage}}
  \caption{Two MST games with $n=3$ players for insights into optimal almost core solutions.}
\end{figure}

\subsection{Two-Approximation Algorithm}
We next propose the following polynomial time algorithm to compute an approximately optimal almost core allocation for problem~\eqref{eq:MSTAlmostCoreMaxProblem}.
For notational convenience, let us define for all $K\subset N$ 
\[
N_{-K} \coloneqq N\setminus K\,,
\]
and write $N_{-i}$ instead of $N_{-\{i\}}$.

\newcommand{\Nkl}{N \setminus \{k,\ell\}}

\begin{algorithm}
  \DontPrintSemicolon
  \SetAlgoLined
  \KwIn{Agents $N$, edge set $E$ of complete graph on $N \cup \{0\}$ and edge weights $w : E \to \Rnonneg$.}
  \KwOut{Almost core allocation $x$.}
  \caption{Approximation algorithm for the almost core maximization problem \eqref{eq:MSTAlmostCoreMaxProblem} for MST games}
  \label{algo_mst_apx}
  Initialize $I_0 \coloneqq \{0\}$ and $T \coloneqq \emptyset$. \;
  \For{$k=1,2,\dotsc,n$}
  {%
    Let $i \in I_{k-1}$, $j \in N \setminus I_{k-1}$ with minimum $w(i,j)$ (among those $i,j$). \;
    Let $I_k \coloneqq I_{k-1} \cup \{j\}$ and augment the tree $T \coloneqq T \cup \{ \{i,j\} \}$. \;
    Assign agent $j$ the cost share $x_j \coloneqq w(i,j)$. \label{algo_mst_apx_assign}
  }
  Let $\ell \in I_n \setminus I_{n-1}$ be the last assigned agent. \label{algo_mst_apx_last_player} \;
  Update agent $\ell$'s cost share $x_\ell \coloneqq \min\limits_{k \in N_{-\ell}} \{ c(N_{-k}) - x(\Nkl) \}$ \label{algo_mst_apx_update}.
\end{algorithm}
\bigskip

The backbone of Algorithm~\ref{algo_mst_apx} is effectively Prim's algorithm to compute a minimum cost spanning tree~\cite{Prim57}.
The additional line~\ref{algo_mst_apx_assign} yields the core allocation by Granot and Huberman~\cite{Granot1981}, which we extend by adding lines~\ref{algo_mst_apx_last_player} and~\ref{algo_mst_apx_update}.

\DeclareDocumentCommand{\xalg}{}{x^{\text{ALG}}}
\DeclareDocumentCommand{\xopt}{}{x^{\text{OPT}}}

Let us first collect some basic properties of Algorithm~\ref{algo_mst_apx}. Henceforth, we assume w.l.o.g.\ that the agents get assigned their cost shares in the order $1,\dots,n$ (so that $\ell=n$ in lines~\ref{algo_mst_apx_last_player} and \ref{algo_mst_apx_update}). We denote by $\xalg$ a solution computed by Algorithm~\ref{algo_mst_apx}. 
 
\begin{lemma}\label{lem:up_to_en_minus_one}
We have that $\xalg(I_k)=c(I_k)$ for all $k=1,\dots, n-1$, and for all $S\subseteq\{1,\dots,n-1\}$ we have $\xalg(S)\le c(S)$.
\end{lemma}
\begin{proof}
The first claim follows directly because Algorithm~\ref{algo_mst_apx} equals Prim's algorithm to compute a minimum cost spanning tree on the vertex set $\{0,\dots,n-1\}$, 
and $\xalg(I_k)$ equals the cost of the minimum cost spanning tree on vertex set $\{1,\dots,k\}$, Hence by correctness of Prim's algorithm~\cite{Prim57}, $\xalg(I_k)=c(I_k)$. The second claim follows by~\cite[Thm.~3]{Granot1981},
since the cost allocation for agents $\{1,\dots,n-1\}$ is the same as in~\cite{Granot1981}.
\end{proof}

\begin{lemma}
\label{lem:core_lemma_feasibility}
Suppose $\xalg(S) > c(S)$ for some set $S$ with $n\in S \subsetneqq N$.
Then there is a superset $T \supseteq S$ with $|T| = n-1$ such that $\xalg(T) > c(T)$.
\end{lemma}
\begin{proof}
Recall the agents got assigned their cost shares in order $1,\dots,n$. Define $k\coloneqq\max\{i\ |\ i\notin S\}$ to be the largest index of a agent not in $S$. Let $i_1,\dots,i_\ell$ be the set of agents so that $N_{-k}=S\cup\{i_1,\dots,i_\ell\}$ and  w.l.o.g.\ $i_1<\cdots < i_\ell$. 
We show that $\xalg(S) > c(S)$ implies $\xalg(S\cup\{i_1\}) > c(S\cup\{i_1\})$. Then repeating the same argument, we inductively arrive at the conclusion that $\xalg(N_{-k})>c(N_{-k})$. So observe that
\[
\xalg(S\cup\{i_1\})=\xalg(S)+x_{i_1} > c(S) +x_{i_1}\,,
\]
and $c(S)$ is the cost of a minimum cost spanning tree for $S$, call it $\textup{MST}(S)$. Moreover, as $i_1\neq n$, $x_{i_1}$ is the cost of the edge, call it $e$, that the algorithm used to connect agent $i_1$. We claim that $\textup{MST}(S)\cup \{e\}$ is a tree spanning
vertices $S\cup\{0,i_1\}$, hence $c(S) +x_{i_1}$ is the cost of some tree spanning 
$S\cup\{0,i_1\}$. Then, as required we get
\[
\xalg(S\cup\{i_1\}) > c(S) +x_{i_1} \ge c(S\cup\{i_1\})\,,
\]
because $c(S\cup\{i_1\})$ is the cost of a \emph{minimum cost} tree spanning $S\cup\{0,i_1\}$.
If $\textup{MST}(S)\cup \{e\}$ was not a spanning tree for vertices $S\cup\{0,i_1\}$,
then edge $e$ would connect $i_1$ to some vertex outside $S\cup\{0\}$, 
but this contradicts the choice of $i_1$ as the vertex outside $S$ with smallest index.
\end{proof}

\begin{lemma}
\label{lem:nonnegativity_xalg}
We have $\xalg \geq 0$.
\end{lemma}
\begin{proof}
Recall that in minimum cost spanning tree games~\cite{ClausK1973,Granot1981}, the weight of edges are non-negative. 
Since Algorithm~\ref{algo_mst_apx} computes the allocation for agents in line~\ref{algo_mst_apx_assign} by the edge weight of the first edge on the unique path to $0$, 
there is $\xalg_{k}
\geq 0$ for all $k=1,2,\cdots,n-1$. So we only need to argue about $\xalg_n$. To that end, note that an equivalent definition of $\xalg_n$ in line~\ref{algo_mst_apx_update} of the algorithm is
\begin{equation}
\label{eq:helper_in_nonnegativity}
\text{max.}\ x_n\ \text{s.t.}\ x_n \le c(N_{-k})-\xalg(N\setminus\{k,n\})\ \text{for all}\ k=1,\dots,n-1\,. 
 \end{equation}
We claim that $\tilde{x}_n\coloneqq c(N)-c(N_{-n})\ge 0$
is a feasible solution to this maximization problem, hence the actual value of $\xalg_n$ after the update in 
line~\ref{algo_mst_apx_update} 
can only be larger, and therefore in particular it is non-negative.
First, note that indeed, $\tilde{x}_n \ge 0$, as this is the cost of the last edge that Prim's algorithm uses to connect the final vertex $n$ to the minimum cost spanning tree.
That $\tilde{x}_n$ is feasible in \eqref{eq:helper_in_nonnegativity} follows from the fact that $\tilde{x}_n$ is the cost share that is assigned to agent $n$ in the core allocation of~\cite{Granot1981}.
Indeed, letting $\tilde{x}$ be equal to $x$ except for $\tilde{x}_n=c(N)-c(N_{-n})$, we have that $\tilde{x}$ is precisely the cost allocation as proposed in~\cite{Granot1981}. 
By the fact that this yields a core allocation, we have that $\tilde{x}(S)\le c(S)$ for all $S\subseteq N$, so in particular for all $k=1,\dots, n-1$,
\[
\tilde{x}_n + \xalg(N\setminus \{k,n\})=\tilde{x}(N_{-k})\le c(N_{-k})\,,
\]
and hence the claim follows.
\end{proof}

\begin{theorem}
\label{thm:2_approximation}
  Algorithm~\ref{algo_mst_apx} is a 2-approximation for the almost core maximization problem \eqref{eq:MSTAlmostCoreMaxProblem} for minimum cost spanning tree games, and this performance bound  is tight for Algorithm~\ref{algo_mst_apx}.
\end{theorem}

\begin{proof}
Denote by $\xalg$ a solution by Algorithm~\ref{algo_mst_apx}.
We first argue that Algorithm~\ref{algo_mst_apx} yields a feasible solution.
For $S\not\ni n$, this follows from \cref{lem:up_to_en_minus_one}.
For $S\ni n$, assume  $x(S)> c(S)$.
Then  \cref{lem:core_lemma_feasibility} yields that there exists some $N_{-k}\ni n$ with $\xalg(N_{-k})> c(N_{-k})$.
However by definition of $x_n$ in line~\ref{algo_mst_apx_update} of the algorithm, we have for all $k=1,\dots,n-1$
\[
\xalg_n \le c(N_{-k}) - \xalg(N \setminus \{k,n\})\,,
\]
which yields a contradiction to $\xalg(N_{-k})> c(N_{-k})$.

To show that the performance guarantee is indeed 2, let $\xopt$ be some optimal solution to the almost core maximization problem.
Let $k^\star \in N_{-n}$ be the index for which the minimum in line~\ref{algo_mst_apx_update} is attained.
Observe that $\xalg_n$ is updated such that $\xalg(N_{-k^\star}) = c(N_{-k^\star})$ holds.
Then by non-negativity of $\xopt$ and because of \cref{lem:nonnegativity_xalg},
\[
\xopt_n  \le \xopt(N_{-k^\star}) \le c(N_{-k^\star}) = \xalg(N_{-k^\star}) \  \le \xalg(N)\,.
\]
Moreover, by definition of $\xalg$, we have  $\xalg(N_{-n}) = c(N_{-n})$, and by \cref{lem:nonnegativity_xalg}, 
\[
\xopt(N_{-n})\le c(N_{-n}) = \xalg(N_{-n}) \le \xalg(N)\,.
\]
Hence we get $\xopt(N) = \xopt_n + \xopt(N_{-n})\le 2 \xalg(N)$.

To see that the performance bound 2 is tight for Algorithm~\ref{algo_mst_apx}, consider the instance in Figure~\ref{fig:the_tight_bound_2}.
\begin{figure}[thb]
\subfloat[MST game showing that the analysis of Algorithm~\ref{algo_mst_apx} cannot be improved.
  \label{fig:the_tight_bound_2}]
{\begin{minipage}{0.47\textwidth}\centering
  \centering
 \begin{tikzpicture}
\node[shape=rectangle,draw=black,fill=black] at (0,0) {};
\draw[fill=black] (4,0) circle (3pt);
\draw[fill=black] (2,1) circle (3pt);
\draw[fill=black] (2,3) circle (3pt);
\node at (-0.3,0) {0};
\node at (4.3,0) {2};
\node at (2.25,1.15) {1};
\node at (2,3.3) {3};
\node at (2,0.2) {$2$};
\node at (0.8,1.7) {$2$};
\node at (1.2,0.8) {$1  $};
\node at (2.8,0.8) {$0$};
\node at (1.8,1.8) {$\varepsilon$};
\node at (3.1,1.7) {$0$};
\draw[thick] (0,0) -- (4,0) -- (2,1) -- (0,0) -- (2,3) -- (4,0) -- (2,1) -- (2,3);
\end{tikzpicture}
\end{minipage}
}
  \hfill
  \subfloat[MST game showing that  Algorithm~\ref{algo_mst_apx} need not compute an element of $AC_{(N,\bar{c})}$.
  \label{fig:monotonized_game_example}]
{\begin{minipage}{0.47\textwidth}\centering
  \centering
  \begin{tikzpicture}
\node[shape=rectangle,draw=black,fill=black] at (0,0) {};
\draw[fill=black] (4,0) circle (3pt);
\draw[fill=black] (2,1) circle (3pt);
\draw[fill=black] (2,3) circle (3pt);
\node at (-0.3,0) {0};
\node at (4.3,0) {2};
\node at (2.25,1.15) {1};
\node at (2,3.3) {3};
\node at (2,0.2) {$1$};
\node at (0.8,1.7) {$1$};
\node at (1.2,0.8) {$1$};
\node at (2.8,0.8) {$0$};
\node at (1.8,1.8) {$0$};
\node at (3.1,1.7) {$1$};
\draw[thick] (0,0) -- (4,0) -- (2,1) -- (0,0) -- (2,3) -- (4,0) -- (2,1) -- (2,3);
\end{tikzpicture}
\end{minipage}
}
\caption{Two MST games with $n=3$ players for the analysis of Algorithm~\ref{algo_mst_apx}.}
\end{figure}
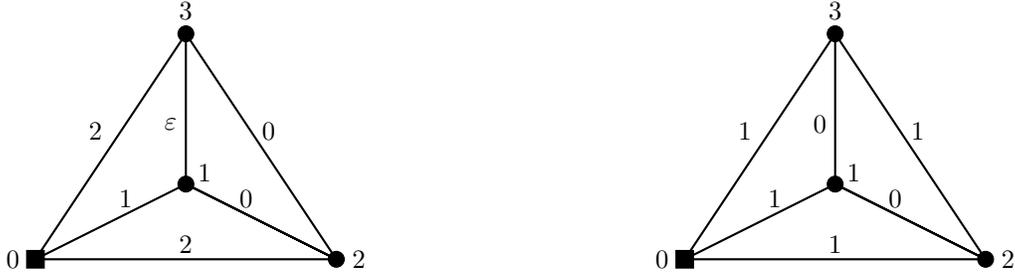
Here, Algorithm~\ref{algo_mst_apx} computes the solution $\xalg=(1,0,\varepsilon)$ with value $1+\varepsilon$, as the order in which agents get assigned their cost shares is $1,2,3$, and in line~\ref{algo_mst_apx_update} of the algorithm we get $\xalg_3= c(\{1,3\})-x_1=(1+\varepsilon)-1=\varepsilon$.
An almost core optimum solution would be $\xopt = (0,1,1)$ with value $2$.
\end{proof}

Even though Theorem~\ref{thm:hardness_MSt_AC_nonneg} suggests that the non-negativity requirement $x \geq 0$ is irrelevant for optimization, it is important for \cref{thm:2_approximation}.
Without it, so allowing $x_i<0$ for some agents $i$, the above algorithm does not provide an approximation guarantee in general.
To see that, consider again the instance given in Figure~\ref{fig:subsidy_example}, and observe that Algorithm~\ref{algo_mst_apx}  yields a cost allocation $\xalg=(0,0,0)$, while $x=(-k,k,k)$ is a feasible solution for the almost core.

It remains to remark that Algorithm~\ref{algo_mst_apx} does generally \emph{not} compute an allocation in the almost core of the corresponding monotonized game $(N,\bar{c})$, as can be seen {for} the instance in Figure~\ref{fig:monotonized_game_example}.
Here, we have that $c(S)=1$ for all $S\subseteq N$ except for $c(\{2,3\})=2$.
An optimal almost core allocation is $x = (0,1,1)$, and depending on how ties are broken, Algorithm~\ref{algo_mst_apx} yields $\xalg = (0,1,1)$ or $\xalg = (1,0,0)$.
The monotonized game has $\bar{c}(S)=1$ for all $S\subseteq N$, and then an optimal almost core allocation is $\bar{x}=(\frac12,\frac12,\frac12)$.
Note that this example also shows that Proposition~\ref{thm_last_monotone_approximation} is tight (for $n=3$), as $c(N)=1$.

\section{Conclusions}
\label{sec_conclusions}
In the literature, one also finds minimum cost spanning tree games defined as \emph{profit sharing} games, where one defines the value of a coalition $S$ by the cost savings that it can realize in comparison to the
situation where all agents in $S$ connect directly to the source,
\[
v(S)\coloneqq\sum_{i\in S}c(\{i\})-c(S)\,.
\]
Then the core constraints, for profit shares $x^v\in\R^n$,  are $x^v(S) \ge v(S)$.
It is not hard to see that all our results also hold for that version of the problem via the simple transformation $x^v_i\coloneqq c(\{i\})-x_i$.
In particular, note that for value games all feasible solutions $x^v$ are non-negative, as core stability requires that $x^v_i \geq v(\{i\}) \geq 0$.
Our results imply $\cplxNP$-hardness for computation of minimum profit shares that are coalitionally stable, and the corresponding profit  version of Algorithm~\ref{algo_mst_apx} can be shown to yield a 2-approximation, which also can be shown to be tight.

We collect some open problems which we believe are interesting.
First, we would like to gain more insight into the computational complexity for the almost core problem \eqref{eq:main_AC_problem}, also for other classes of games.
In particular, note that the computation of an optimal almost core allocation for submodular cost functions as of Corollary~\ref{cor:submodular_ac_problem} relies on the equivalence of separation and optimization. Since the computation of a core element can be done by Edmonds' greedy algorithm, it is conceivable that dropping just one single inequality from that polymatroid still allows for a  combinatorial algorithm, without the need to resort to the Ellipsoid method.
Moreover, we gave a 2-approximation for cost MST games under the additional assumption that subsidies are not allowed. 
It would be interesting to extend this result to the general, unconstrained case, or show that this is not possible.
Also giving lower bounds on the approximability does seem plausible, as the ``hard cases'' for maximizing shareable costs are those where a minimum cost spanning tree exists which is a (Hamiltonian) path.

Finally, both in general and for MST games one could define a more general class of problems in the spirit of cooperative games with restricted coalition formation, 
by defining  a (downward-closed) set system that describes all those subsets of agents that are able to cooperate and hence have access to an outside option, while all other subsets do not have that option.
This is the same idea as that of restricted coalition formation by Myerson~\cite{Myerson1977} or Chalkiadakis et al.~\cite{Chalkiadakis2016}.
The almost core as studied in this paper is the special case where this set system is particularly simple, namely the $(n-1)$-uniform matroid.

\section*{Acknowledgements}
Rong Zou and Boyue Lin acknowledge the support of the China Scholarship Council (Grants No.\ 202006290073, 202106290010). The authors also thank the anonymous reviewers of an earlier draft of this paper for some constructive comments.



\bibliographystyle{acm} 
\bibliography{almost-core}


\end{document}